\documentclass[conference,10pt]{IEEEtran}
\usepackage[utf8]{inputenc}
\def\BibTeX{{\rm B\kern-.05em{\sc i\kern-.025em b}\kern-.08em
    T\kern-.1667em\lower.7ex\hbox{E}\kern-.125emX}}

\usepackage{graphicx}
\usepackage{color}
\usepackage{cite}
\usepackage{amsthm}
\usepackage[caption=false]{subfig}
\usepackage{amssymb}
\usepackage{comment}
\usepackage{mathtools}
\usepackage{tabulary}
\usepackage{acronym}
\usepackage{upgreek}
\usepackage{algorithm}
\usepackage{algpseudocode}
\usepackage{placeins}
\usepackage[super]{nth}
\usepackage{dsfont}
\usepackage{hyperref}

\newtheorem{Lemma}{Lemma}

\include{acronyms}

\title{How Do Microstrip Losses Impact Near-Field Beam Depth in Dynamic Metasurface Antennas?}

\author{
    \IEEEauthorblockN{
        Panagiotis Gavriilidis and George C. Alexandropoulos
    }
    \IEEEauthorblockA{
 Department of Informatics and Telecommunications,
        National and Kapodistrian University of Athens, Greece\\
        e-mails: \{pangavr, alexandg\}@di.uoa.gr
    }
}

\begin{document}

\maketitle

\begin{abstract}

The convergence of eXtremely Large (XL) antenna arrays and high-frequency bands in future wireless networks will inevitably give rise to near-field communications, localization, and sensing. Dynamic Metasurface Antennas (DMAs) have emerged as a key enabler of the XL Multiple-Input Multiple-Output (MIMO) paradigm, leveraging reconfigurable metamaterials to support large antenna arrays. However, DMAs are inherently lossy due to propagation losses in the microstrip lines and radiative losses from the metamaterial elements, which reduce their gain and alter their beamforming characteristics compared to a lossless aperture. In this paper, we address the gap in understanding how DMA losses affect near-field beamforming performance, by deriving novel analytical expressions for the beamforming gain of DMAs under misalignments between the focusing position and the intended user's position in 3D space. Additionally, we derive beam depth limits for varying attenuation conditions, from lossless to extreme attenuation, offering insights into the impact of losses on DMA near-field performance.
\end{abstract}
\begin{IEEEkeywords}
Near field, beam focusing, beam depth, dynamic metasurface antennas, misalignment, XL MIMO.
\end{IEEEkeywords}

\section{Introduction} \label{Sec:Intro}
\let\thefootnote\relax\footnotetext{This work has been supported by the SNS JU project TERRAMETA under the EU's Horizon Europe research and innovation program under Grant Agreement No 101097101, including top-up funding by UKRI under the UK government's Horizon Europe funding guarantee.}

Future Sixth-Generation (6G) wireless networks aim to support demanding applications, such as ultra-high data rate communications, precise localization, and environmental sensing, which impose stringent requirements on spectral efficiency, latency, and energy consumption \cite{HMIMO_survey_et_al}. To meet these demands, millimeter-wave and higher frequency bands, combined with eXtremely Large (XL) aperture antenna arrays, are considered as key enablers. The synergy of these technologies, not only facilitates large-scale spatial multiplexing, but also ensures that the high beamforming gain of XL Multiple-Input Multiple-Output (MIMO) systems can compensate for the severe pathloss at high frequencies. This combination naturally leads to the emergence of near-field communications, which fundamentally differ from conventional far-field models that rely primarily on angular channel representations \cite{NF_tutorial}.

In the near-field regime, the spherical nature of wavefronts needs to be accounted for, as the conventional plane-wave approximation is no longer valid. This shift unlocks new degrees of freedom, enabling the service of users positioned at the same angular direction, thereby enhancing spatial multiplexing and interference management. However, the orthogonality of near-field beam focusing vectors
is fundamentally constrained by the user’s distance from the antenna array and the array’s physical dimensions~\cite{NF_Beam_tracking}. The near-field beamforming characteristics, and particularly, the depth of focus for planar arrays was introduced in~\cite{Bjornson_dist} 
for a user positioned along the normal vector from the array center. An analytical approximation of the orthogonality between two beam focusing vectors for linear arrays was derived in~\cite{ULA_NF_analysis} and later extended to planar arrays in~\cite{SDMA_vs_LDMA}, where a spherical-domain codebook for beamforming was also proposed. More recently,~\cite{NF_Beam_tracking} examined the near-field beam depth limits for a lossless Dynamic Metasurface Antenna (DMA), focusing on scenarios where the user remains within a constant 2D plane, and proposed a dynamic near-field sampling approach for localization, adapting to the beamforming resolution that the Base Station (BS) aims to achieve.

DMAs represent a promising alternative to conventional hybrid analog/digital MIMO architectures, leveraging reconfigurable metamaterials to enable highly flexible and compact implementations of XL MIMO arrays~\cite{DMA_Magazine}. Unlike hybrid MIMO systems, which rely on power-hungry phase shifters and power splitters, DMAs operate via microstrip lines, where embedded metamaterial elements dynamically modify their electromagnetic properties to achieve high-resolution beamforming, while maintaining reduced power consumption. These advantageous characteristics have recently fueled growing research interest in the field~\cite{gavras2023duplex, DMA_near_field_channel, DMA_DL_beamforming_TWC22}. However, due to their intrinsic structure, DMAs introduce additional challenges, particularly in terms of propagation and radiative losses along the microstrip lines, which can degrade their beamforming performance. The authors in~\cite{DMA_effective_aperture} studied the far-field beamforming characteristics of such waveguide-fed metasurfaces and derived an analytical expression for the attenuation coefficient within the microstrips.
Motivated by the growing research interest in DMAs and their potential to unlock near-field communications, this paper focuses on quantifying the effects of DMA losses on near-field beamforming performance. We first derive analytical approximations for the deterioration of the near-field beamforming gain due to misalignments in the range coordinate. Subsequently, we establish the beam depth limits, for which a certain percentage of the maximum beamforming gain is lost, analytically assessing the impact of losses on both.

\vspace{-2.5mm}
\section{Modeling of Lossy DMAs}\label{Sec: System Model}

We consider a DMA-equipped BS communicating with a single-antenna User Equipment (UE). The DMA consists of $N_m$ microstrips, each fed by a single RF chain and comprising $N_e$ metamaterial elements, resulting essentially in a total of $N = N_m N_e$ antenna elements. Furthermore, let $s$ denote the unit-power complex-valued information symbol transmitted via the linearly beamformed vector:
\(\mathbf{x} = \bar{\mathbf{Q}} \mathbf{v} s \in \mathbb{C}^{N \times 1}\),
where $\bar{\mathbf{Q}} \in \mathbb{C}^{N \times N_m}$ and $\mathbf{v} \in \mathbb{C}^{N_m \times 1}$ represent the DMA’s analog and digital beamformers, respectively. The latter is subject to a power restriction $\|\mathbf{v}\|_2^2 \leq P_b$, where $P_b$ denotes the maximum transmit power budget of the BS. The DMA analog beamformer is defined as
\(\bar{\mathbf{Q}} \triangleq \mathbf{P}_m \mathbf{Q}\),
where the diagonal matrix $\mathbf{P}_m \in \mathbb{C}^{N \times N}$ models the propagation of the signal inside each microstrip, while $\mathbf{Q} \in \mathbb{C}^{N \times N_m}$ contains the tunable responses of the metamaterial elements that follow the Lorentzian-constrained profile, according to which, for each $n$-th element (\(n=0,1,\ldots,(N_e-1)\)) in each $i$-th microstrip (\(i=0,1,\ldots,(N_m-1)\)), the following holds:
\begin{equation}\label{eq:Q_matrix}
[\mathbf{Q}]_{iN_{e}+(n+1),j}=\begin{cases} q_{i,n}\in \mathcal{Q},&i=j\\0,&i\neq j\end{cases}
\end{equation}
with $\mathcal{Q}\triangleq\left\{0.5\left(\jmath + e^{\jmath\phi}\right) | \phi \in \left[0,2\pi\right]\right\}$ \cite[eq.~(30)]{DMA_effective_aperture}. In this paper, we model the amplitude and phase distortion due to signal propagation within the microstrip, before reaching each metamaterial element, as follows:
\begin{equation}
[\mathbf{P}_m]_{iN_e+ (n+1), iN_e+(n+1)} =  \frac{\exp \left( -\alpha \rho_{i,n} - \jmath \beta \rho_{i,n} \right)}{\sqrt{N_e}},
\end{equation}
where $\beta$ and \(\alpha\) are the wavenumber and attenuation coefficients of the microstrip, respectively, while $\rho_{i,n}$ represents the distance of the $n$-th element in the $i$-th microstrip from the input port, which is given by \(\rho_{i,n} = n d_e\). 
Unlike conventional downlink modeling of DMAs in~\cite{gavras2023duplex,DMA_near_field_channel,DMA_DL_beamforming_TWC22}, we have introduced the normalization term \((\sqrt{N_e})^{-1}\) to ensure compliance with the power conservation principle. Note that, without this term, the radiated power \(||\mathbf{\bar{Q}v}||_2^2\) would be artificially amplified by a factor of \(N_e\), since each column of \(\mathbf{\bar{Q}}\) would exhibit a squared norm on the order of \(N_e\). 
From a physical standpoint, this normalization assumes that each metamaterial element within a microstrip experiences the same overall attenuation due to radiative losses, that is, power escaping the transmission line as it is radiated by the other elements, regardless of its position relative to the input port. In reality, however, radiation losses affect the metamaterial elements non-uniformly \cite{DMA_effective_aperture}. A more detailed analysis could account for the spatially varying effects of radiation losses, but modeling this in a way that ensures compliance with the energy conservation principle is beyond the scope of this paper. Herein, the attenuation coefficient \(\alpha\) is actually assumed to primarily capture the conductor and dielectric losses in the microstrip line, as modeled in \cite[eqs.~(3.198) and (3.199)]{pozar2012microwave}. However, since our derivations remain parametric in \(\alpha\), they are applicable to a broader range of scenarios beyond these loss mechanisms.


The BS is assumed to be centered at the origin $(0,0,0)$, with the DMA surface aligned with the $zy$-plane. Considering a Line-of-Sight (LoS) channel between the BS and the UE, the baseband received signal at the UE is given by:
\begin{equation}
y \triangleq \mathbf{h}_{\rm LoS}^{\rm H} \mathbf{x} + n,
\end{equation}
where \(n\) is the Additive White Gaussian Noise (AWGN) and $\mathbf{h}_{\rm LoS}\in \mathbb{C}^N$ is the LoS channel vector defined as:
\([\mathbf{h}_{\rm LoS}]_{iN_e + (n+1)} \triangleq \frac{\lambda}{4\pi r_{i,n}} e^{ -\jmath \frac{2\pi}{\lambda} r_{i,n} }.
\)
Here, $r_{i,n}$ represents the distance from the $n$-th element of the $i$-th microstrip to the UE. Assuming the UE is located at position $(r,\phi,\theta)$ and using the Fresnel approximation \cite{NF_tutorial}, the pathloss term over the DMA's aperture can be considered as constant, thus, the LoS channel can be reformulated as \(\mathbf{h}_{\rm LoS}=\frac{\lambda}{4 \pi r_{0}}\mathbf{a}(r,\phi,\theta)\) with \([\mathbf{a}(r,\phi,\theta)]_{iN_e+ (n+1)}\triangleq e^{-j \frac{2 \pi}{\lambda} r_{i,n}} \) being the focusing vector, \(r_0\) the distance from the DMA's center, and \(r_{i,n}\) given in~\eqref{eq: r_in distance from BS}.

\begin{figure*}
\begin{align}\label{eq: r_in distance from BS}
r_{i,n} = & \,\sqrt{(r \sin(\theta) \cos(\phi))^2 + (r \cos(\theta) - n_z d_e)^2 + (r \sin(\theta) \sin(\phi) - i_y d_m)^2}\\
r_{i,n} \cong&\, r - n_z d_e \cos(\theta) - i_y d_m \sin(\theta)\sin(\phi) + \frac{n^2 d^2}{2 r} (1 - \cos^2 (\theta)) + \frac{i_y^2 d_m^2}{2 r} (1 - \sin^2(\theta) \sin^2(\phi)) - \frac{n_z i_y d_e d_m \cos(\theta) \sin(\theta) \sin(\phi)}{r}\nonumber\\
n_z\triangleq & \, n- 0.5(N_e-1),\nonumber
\; i_y\triangleq \, i- 0.5(N_m-1)
\end{align}
\hrule
\vspace{-5mm}
\end{figure*}

To maximize the signal-to-noise ratio at a specific UE position $(r,\phi,\theta)$, the analog beamformer at the DMA is set to align with the phase of $\mathbf{a}^{\rm H}(r,\phi,\theta)$~\cite[Lemma 1]{DMA_loc_Nir}. Furthermore, according to \cite[Proposition~1]{NF_Beam_tracking}, if this selection is made for the analog beamformer, then the optimized digital beamformer reduces to a vector of ones scaled appropriately to meet the power constraint. 
Therefore, the optimal hybrid analog and digital DMA configuration is given as:
\begin{equation}\label{eq: Hybrid Analog and Digital Config for Los focusing}
\mathbf{\bar{Q}}\mathbf{v}= 0.5\sqrt{\frac{P_b}{N}}\left(\mathbf{a}(r,\phi,\theta)\odot e^{-\alpha \boldsymbol{\rho}}  + \jmath \,e^{-(\alpha  +\jmath \beta) \boldsymbol{\rho}}\right),
\end{equation}
where the symbol ``\(\odot\)'' refers to the element-wise multiplication of the involved vectors, and the term \(e^{-(\alpha +\jmath \beta) \boldsymbol{\rho}}\) models waveguide propagation within the microstrips and does not contribute to the beamforming gain. This expression represents the focusing vector created by the DMA. Due to attenuation within the microstrip lines, instead of the actual focusing vector towards the UE, the DMA essentially creates a decaying focusing vector which we denote as \(\mathbf{a}_{\rm DMA} \in \mathbb{C}^N\) and its elements are given as \([\mathbf{a}_{\rm DMA}(r,\phi,\theta)]_{iN_e+(n+1)} \triangleq \exp\left(\jmath \angle [\mathbf{h}_{\rm LoS}]_{iN_e+(n+1)}-\alpha \rho_{i,n}\right)\).

Let us define the beamforming gain as \( G \triangleq |\mathbf{a}^{\rm \rm H}(r,\phi,\theta)\mathbf{\bar{Q}}\mathbf{v} |^{2}\), which corresponds to the gain achieved at the UE's position (excluding the pathloss). As expected, the maximum gain \(G_{\rm opt}\) occurs for the case of a perfect alignment between the BS focusing position and the UE's position. Incorporating \eqref{eq: Hybrid Analog and Digital Config for Los focusing}, this gain is given as follows:
\begin{align}\label{eq: Gain for perfect alignment}
    & G_{\rm opt} = \frac{P_b}{4N} \left|N_m\sum_{n=0}^{(N_e-1)}e^{-\alpha n d_e} \right|^2\,\,\Rightarrow\nonumber\\
    & G_{\rm opt} = \frac{P_b N_m}{4N_e}\left( \frac{e^{-\alpha d_e N_e}-1}{e^{-\alpha d_e}-1}\right)^2.
\end{align}
For the lossless case, i.e., when \(\alpha=0\), the following limit is deduced: \(\lim_{\alpha \to 0}G_{\rm opt} = 0.25 P_b N_m N_e\), since \(\lim_{\alpha \to 0} \frac{e^{-\alpha d_e N_e}-1}{e^{-\alpha d_e}-1} = N_e\) holds.
Let us now define the parameter \(\eta \triangleq 1/N_e \frac{e^{-\alpha d_e N_e}-1}{e^{-\alpha d_e}-1}\), so that \(\eta N_e\) is the number of ``effective'' DMA elements, i.e., the number of elements that would achieve the same maximum gain if there were no losses. Clearly, the maximum gain is given by \(G_{\rm opt} = 0.25 P_b \eta^2 N\).

\section{Near-Field Beam Depth Analysis}
\subsection{Beamforming Gain for Range Mismatches}\label{subsec: Near field beamforming gain analysis}
When there is a mismatch in the range coordinate between the BS focus point and the actual UE position, we can compute the gain as:
\(
G(\hat{r},{\phi},{\theta}) \triangleq 0.25 \frac{P_b}{N} |\mathbf{a}^{\rm H}(r,\phi,\theta)\mathbf{a}_{\rm DMA}(\hat{r},{\phi},{\theta}) |^2.
\)
Furthermore, the relative beamforming gain, i.e., the achieved gain divided by the maximum gain, is expressed as follows:
\begin{equation}
\bar{G}(\hat{r},\hat{\phi},\hat{\theta}) \triangleq \frac{G(\hat{r},\hat{\phi},\hat{\theta})}{G_{\rm opt}} = \frac{|\mathbf{a}^{\rm H}(r,\phi,\theta)\mathbf{a}_{\rm DMA}(\hat{r},\hat{\phi},\hat{\theta}) |^{2}}{\eta^2 N^2 }.
\end{equation}
In this paper, we investigate the near-field beamforming limits arising from range mismatches, focusing specifically on the corresponding beam depth limits, where a certain percentage of the maximum beamforming gain is lost. Instead of analyzing the absolute beamforming gain, we study the relative beamforming gain function to examine the directivity characteristics across different microstrip attenuation conditions. In particular, variations in attenuation coefficients lead to different maximum beamforming gains \eqref{eq: Gain for perfect alignment}. By considering the relative beamforming gain, we can systematically determine the beam depth limits at which a specified fraction of the peak gain, regardless of its absolute value, is lost. To facilitate this analysis, we first present the following lemma including a novel analytical approximation for \(\bar{G}(\hat{r},\phi,\theta)\).

\begin{Lemma}\label{prop_Dr}
When the BS beam focuses at the point \((\hat{r}, \phi,\theta)\), where \(\hat{r}=r\pm \Delta r\) with \(\Delta r \geq 0\), the relative beamforming gain can be approximated as~\(\bar{G}(r\pm \Delta r,\phi,\theta) \cong\eta^{-2}\mathcal{K}^2\left(t_z(\pm \Delta r),w\right) \mathcal{D}^2(t_y(\pm \Delta r))\), where:
\begin{align}\label{eq: Analytical function range mismatch z-axis}
 \mathcal{K}(x,w) \triangleq & \frac{\sqrt{\pi}e^{-w}}{2x}\big|{\rm erfi}(e^{\jmath \pi/4}0.5x + e^{-\jmath \pi/4}w/x)- \nonumber\\
&{\rm erfi}(-e^{\jmath \pi/4}0.5x + e^{-\jmath \pi/4}w/x)\big|,\\
\mathcal{D}(x) \triangleq & \frac{1}{x}|C\left(x\right)+\jmath S\left(x\right)| .
\end{align}
In these expressions, \(C(\cdot)\) and \(S(\cdot)\) are the Fresnel functions, \({\rm erfi}(\cdot)\) is the imaginary error function \cite{table_of_integrals}, 
\(t_y(x) \triangleq d_m N_m \sqrt{ 0.5{(1 - \sin^2(\theta)\sin^2(\phi))}/{\lambda}\frac{2|x|}{r^2 + r x}} \), \(t_z(x)=d_e N_e\sqrt{\pi \sin^2(\theta)/\lambda\frac{|x|}{r^2 + r x}}\), and \(w \triangleq 0.5N_e d_e \alpha\).
\end{Lemma}
\begin{proof}
Following \cite[Lemma~3]{SDMA_vs_LDMA}, we omit the bilinear term of \eqref{eq: r_in distance from BS} in the beamforming gain derivations, yielding the formula:
\begin{align}
    & \frac{\mathbf{a}^{\rm H}(r,\phi,\theta)\mathbf{a}_{\rm DMA}(\hat{r},\hat{\phi},\hat{\theta}) }{N } =\nonumber\\
    &\frac{1}{N}\sum_{n_z=-0.5(N_e-1)}^{0.5(N_e-1)}\sum_{i_y=-0.5(N_m-1)}^{0.5(N_m-1)} e^{\jmath \pi/\lambda n_z^2 d_e^2 \sin^2(\theta)\left(\frac{1}{r}-\frac{1}{r \pm \Delta r}\right) }\nonumber\\
    & \times e^{\jmath \pi/\lambda i_y^2 d_m^2(1- \sin^2(\theta)\sin^2(\phi))\left(\frac{1}{r}-\frac{1}{r \pm \Delta r}\right)- (n_z+0.5(N_e-1))d_e\alpha}.\nonumber
\end{align}
The above double sum can be approximated using the Riemmanian sum approximation, by performing an integration in place of the summation. To ease the analysis without compromising the results, we will consider that \(0.5(N_e-1)\cong0.5 N_e\) and \(0.5(N_m-1)\cong0.5 N_m\) \cite{SDMA_vs_LDMA}. The sum of the \(i_y\) index can be approximated following the steps in \cite[Lemma~4]{SDMA_vs_LDMA} and \cite[Lemma~1]{NF_Beam_tracking}. On the other hand, the terms of the sum over \(n_z\) contain a real part in the exponent, and therefore, require separate treatment. Thus, we use the indefinite integral in \cite[page~108 eq.~(13)]{table_of_integrals}, and after setting \(n' = n_z/N_e\), we write:
\begin{align}
    &\frac{e^{-0.5N_ed_e\alpha }}{N_e}\sum_{n'=-0.5}^{0.5} e^{\jmath \frac{\pi}{\lambda} n'^2 N_e^2 d_e^2 \sin^2(\theta)\left(\frac{1}{r}-\frac{1}{r \pm \Delta r}\right) - n'N_e d_e\alpha } \cong\nonumber \\
    & \frac{0.5\sqrt{\pi}e^{-w}}{e^{\pm \jmath \frac{\pi}{4}} t_z(\pm \Delta r)}\bigg({\rm erfi}\!\left(\!e^{\pm \jmath \frac{\pi}{4}}t_z(\pm \Delta r)0.5  + e^{\mp\jmath \frac{\pi}{4}}\frac{w}{t_z(\pm \Delta r)}\!\right) - \nonumber\\
    &{\rm erfi}\!\left(\! -e^{\pm \jmath \frac{\pi}{4}}t_z(\pm \Delta r)0.5  + e^{\mp\jmath \frac{\pi}{4}}\frac{w}{t_z(\pm \Delta r)}\!\right)\!\!\bigg) \label{eq: Analytical erfi} e^{-\jmath \frac{w^2}{t_z^2(\pm \Delta r)}}.
\end{align}
Lastly, by taking the absolute value of \eqref{eq: Analytical erfi} and using the properties \({\rm erfi}^*(x)={\rm erfi}(x^*)\) and \(| {\rm erfi}(x) - {\rm erfi}(y)| = |{\rm erfi}^*(x) - {\rm erfi}^*(y)|\), the proof is concluded.  
\end{proof}

We can further approximate the gain expression from Lemma~1 to reduce its dependency to \(\mathcal{D}(t_y(\pm \Delta r))\), since this will later enable us to derive the beam depth limits for our DMA-based antenna array, in a closed-form fashion.  To do so, we will capitalize on the fact that DMAs consist of few microstrips compared to the numerous metamaterials per microstrip, i.e., typically holds \(N_e\gg N_m\), indicating their advantage to facilitate XL MIMO \cite{DMA_Magazine,NF_Beam_tracking}.
\begin{Lemma}\label{prop reduce dependency to y-axis}
If it holds that \(\frac{(N_m d_m)^2}{(N_e d_e)^3 }\to 0\), then, for finite positive \(\Delta r\), we can write \(\bar{G}(r\pm \Delta r,\phi,\theta) \cong\eta^{-2}\mathcal{K}^2\left(t_z(\pm \Delta r),w\right) \).
\end{Lemma}
\begin{proof}
For less than \(1 \%\) error in the approximation, we set \(\mathcal{D}^2(t_y(\pm \Delta r))^2 \geq 0.99\), which yields \(t_y(\pm \Delta r) \leq 0.46\). Then, we can equivalently write:
\(
    N_m^2d_m^2 \leq \frac{0.46^2 2 r(r\pm \Delta r) \lambda }{ \Delta r (1 - \sin^2(\theta)\sin^2(\phi)) }\Rightarrow  N_m^2d_m^2 \leq\frac{0.46^2 2 r(r\pm \Delta r) \lambda }{ \Delta r  }.
\)
We further make the following assumptions: \textit{i}) \(r\) and \( r\pm \Delta r \geq 0.62 \sqrt{D^3/\lambda}\) (Fresnel distance) \cite{NF_Beam_tracking} with \(D\) being DMA's maximum length; and \textit{ii}) \(D\cong N_e d_e\), then the above inequality can be written as: \( \frac{N_m^2d_m^2}{(N_e d_e)^3} \leq \frac{0.16}{\Delta r}\), which concludes the proof, since \(\frac{N_m^2d_m^2}{(N_e d_e)^3} \to 0\).
\end{proof}
\begin{figure}
    \centering
    \includegraphics[width=\columnwidth]{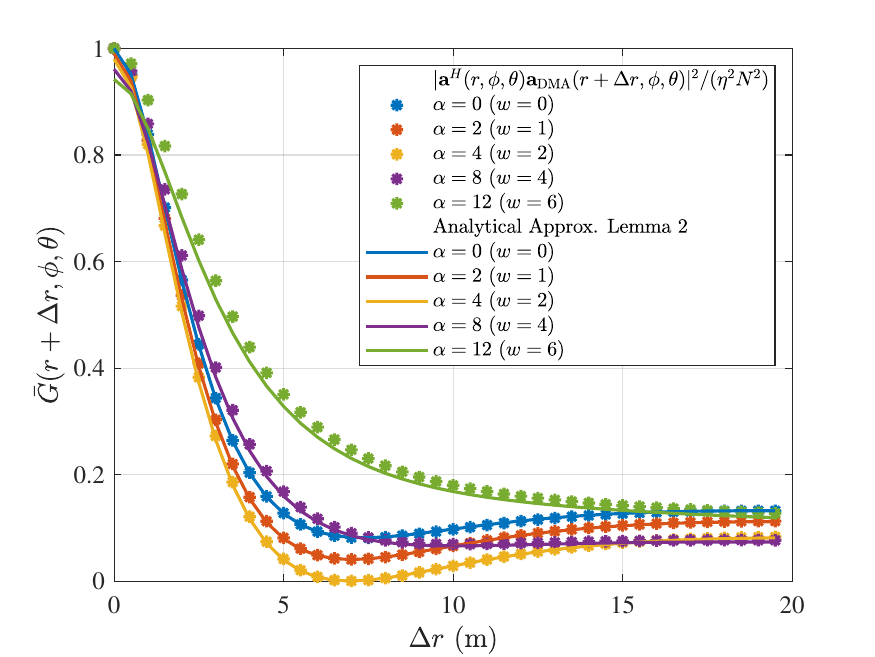}
    \caption{\(\bar{G}(r+\Delta r,\phi,\theta)\) with respect to range mismatch \(\Delta r\) for different attenuation coefficient values \(\alpha=[0,2,4,8,12]\), considering a DMA at the BS with $N_e=200$ and $N_m = 10$, $d_e=d_m=\lambda/2$ with $\lambda = 1$ cm,  \(r=7\) m, \(\phi = \pi/3\), and \(\theta = \pi/2\). The asterisks correspond to the numerical evaluation of \(\frac{|\mathbf{a}^{\rm H}(r,\phi,\theta)\mathbf{a}_{\rm DMA}(\hat{r},\hat{\phi},\hat{\theta})|^2 }{(\eta N)^2 }\), while the solid lines correspond to the analytical approximation of Lemma~2, i.e., \(\eta^{-2} \mathcal{K}^2(t_z(\Delta r), w)\). }
    \label{fig: Lemma validation G(r Delta r)}
    \vspace{-3mm}
\end{figure}

The analytical expression of the relative beamforming gain in Lemma~2 is validated numerically in Fig.~\ref{fig: Lemma validation G(r Delta r)} for varying attenuation coefficients. Specifically, we compare the numerical evaluation of \(\bar{G}(r+\Delta r,\phi,\theta)\) with its analytical counterpart \(\eta^{-2}\mathcal{K}(t_z(\Delta r),w)\), using the simulation parameters provided in the figure's caption. It is shown that the respective two curves exhibit strong agreement, with minor discrepancies for \(\alpha=12\). Notably, for \(\alpha=2\) and \(4\)
or \(w=1\) and \(2\), 
the beamforming gain slope is steeper compared to the lossless case, indicating that the former experiences larger gain loss for the same \(\Delta r\) mismatches. This can appear counterintuitive, as one might expect the reduced effective aperture of the DMA to result in broader beams. In the sequel, we elaborate on this phenomenon and determine up to which point increasing the attenuation makes the beamforming gain slope steeper. 

\vspace{-2mm}
\subsection{Beam Depth Limits versus Attenuation}\label{subsec: Beam depth derivation}
\vspace{-1mm}
Further inspection of Fig.~\ref{fig: Lemma validation G(r Delta r)} reveals that \(\mathcal{K}^2(t_z(\pm \Delta r),w)\) has a decreasing trend with \(\Delta r\) for all \(w\), while numerical evaluations with respect to \(t_z\) (figure omitted) show a strictly decreasing trend up to approximately \(t_z = 4.7\), after which it oscillates while decreasing. Leveraging this trend, we define the depth of focus as the region where the beamforming gain remains above \(100\delta\%\) of its maximum value, i.e., \(\bar{G}(\hat{r},\phi,\theta) \geq \delta,\,\delta \in (0,1),\) for \(\hat{r} \in [r - \Delta^{-}_{\delta}(r,w), r + \Delta^{+}_{\delta}(r,w)]\), where:
\begin{equation}\label{eq: Dr}
  \Delta^{\pm}_{\delta}(r,w) \triangleq r^2\left(\frac{d_e^2 N_e^2 \pi \sin^2(\theta)}{\lambda x_{\delta}^2(w)} \mp r\right)^{-1}.
\end{equation}
Here, \(x_\delta(w)\) denotes the point \(x\) where \(\mathcal{K}^2(x,w) = \delta \mathcal{K}^2(0,w)\), i.e., the location where the gain drops to \(100\delta\%\) of its maximum. The limits in \eqref{eq: Dr} are acquired by setting \(t_z(\pm \Delta r)=x_{\delta}(w)\) and solving for \(\Delta r\). It then follows that, as \(r \to \frac{N_e^2 d_e^2 \pi \sin^2(\theta)}{\lambda x_{\delta}^2(w)}\), \(\Delta^{+}_{\delta}(r,w) \to \infty\) is deduced. This resembles the limiting distance from which there ceases to exist a \(\Delta^{+}_{\delta}(r,w)\) factor, such that, for \(\hat{r} > r + \Delta^{+}_{\delta}(r,w)\), it holds \(\bar{G}(\hat{r},\phi,\theta) < \delta \), since, for \(r\) beyond this limit, the UE progressively transitions into the DMA’s far-field region. On the other hand, the solution \(\Delta^{-}_{\delta}(r,w)\) always exists, indicating the direction of movement towards the BS, i.e., towards its near-field zone. It is also noted that, when \(\sin^2(\theta)=0\), the beam depth cannot be defined since, for that case, it holds \(t_z(\pm \Delta r)=0\, \forall \Delta r\).

It is hard to solve \(\mathcal{K}^2(x,w) = \delta \mathcal{K}^2(0,w)\) analytically, hence, we determine \(x_{\delta}(w)\), for each pair of \(\delta\) and \(w\), numerically. While an exact analytical solution may not be feasible, approximating this function with a simple, yet efficient, model allows us to ease the analysis. In the following, we provide a closed-form expression that captures the essential behavior of the attenuation dependence, making it easier to analyze and gain insights into the impact of the attenuation coefficient on \(x_{\delta}(w)\), and consequently, the beam depth limits in~\eqref{eq: Dr}. 

\begin{figure}
    \centering
    \includegraphics[width= \columnwidth]{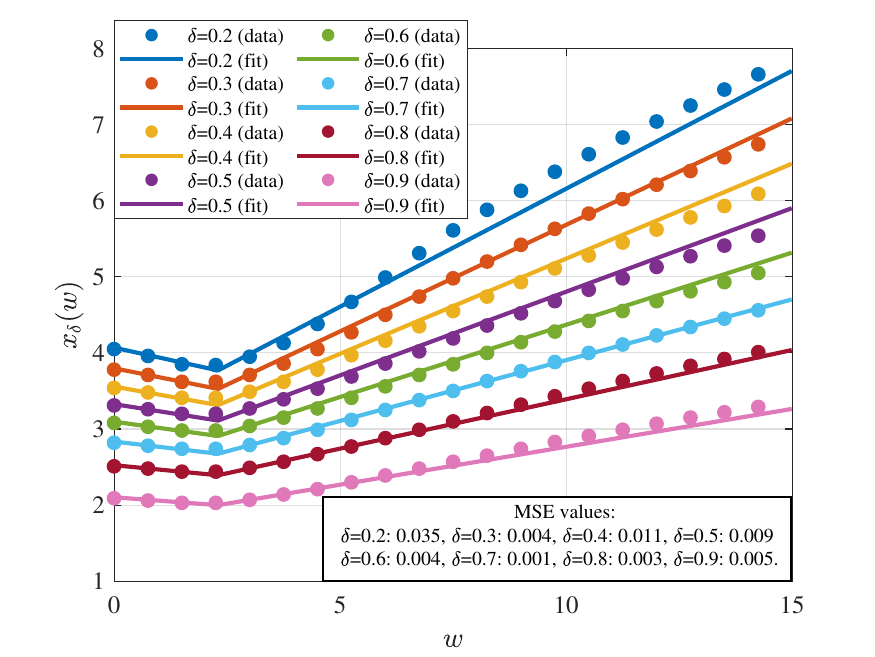}
    \caption{Numerically computed values of \(x_{\delta}(w)\) (circles) compared to the proposed approximation \(\hat{x}_{\delta}(w)\) (solid lines) for \(\delta=0.2:0.1:0.9\). The box in the bottom shows the MSE between \(x_{\delta}(w)\) and \(\hat{x}_{\delta}(w)\), for each \(\delta\) value, computed over the range \(w \in [0,15]\) and averaged over the sampled values.}
    \label{fig:piecewise linear approximation of x_delta}
    \vspace{-3mm}
\end{figure}

In Fig.~\ref{fig:piecewise linear approximation of x_delta}, we plot \(x_{\delta}(w)\) values for \(w \in [0,15]\), as values beyond this range hold no practical relevance. Recall that \(w=0\) corresponds to a lossless microstrip, while \(w=15\) represents an extreme case where the signal undergoes \(\exp(-15)\) (i.e., \(-130\) dB) attenuation at half the microstrip’s length; this is an unrealistic scenario requiring either exceptionally long or highly lossy transmission lines. For reference, a microstrip with a Duroid 5880 substrate \cite[Table~14.1]{balanis2016antenna}, a width of \(\lambda/2\), and operation at \(f_c=30\) GHz exhibits an attenuation coefficient of \(a=0.875\, {\rm m^{-1}}\), meaning its length would need to be \(34\) m long for the signal to reach the attenuation value \(\exp(-15)\) at its midpoint.

Based on our inspection of \(x_{\delta}(w)\)'s behavior in Fig.~\ref{fig:piecewise linear approximation of x_delta}, we propose the following piecewise-linear approximation:  
\begin{equation}\label{eq: piecewise linear fitting of x_delta(w)}
    \hat{x}_{\delta}(w) =
    \begin{cases} 
        x_{\delta}(0) + a_0(\delta) + a_1(\delta) w, & w < w_0, \\
        x_{\delta}(0) + b_0(\delta) + b_1(\delta) w, & w \geq w_0,
    \end{cases}
\end{equation}
where \(a_0(\delta) =0.02 - 0.007\delta\), \(a_1(\delta) =-0.154 + 0.121\delta\), \(b_0(\delta) = -1.186 + 0.963\delta\), \(b_1(\delta) = 0.370 - 0.301\delta\), \(w_0=2.3\), and \(x_{\delta}(0)\) is the point where \(\mathcal{K}^2(x,0)=\delta\mathcal{K}^2(0,0)\) (lossless case). The model parameters \(a_0(\delta),a_1(\delta),b_0(\delta), \) and \(b_1(\delta)\) are linear in \(\delta\) and were optimized via least squares fitting over all sampled instances of \(w\) and \(\delta\): \(
\min \sum_{w \in [0,15]} \sum_{\delta \in [0.2,1)} \|\hat{x}_{\delta}(w) - x_{\delta}(w)\|_2^2 + \sum_{\delta \in [0.2,1)}\|a_0(\delta) + a_1(\delta) w_0 - b_0(\delta) - b_1(\delta) w_0\|_2^2
\), where the latter term enforces continuity at \(w_0\). For \(\delta < 0.2\), \(x_{\delta}(w)\) exhibits highly nonlinear behavior and, while higher-order polynomials could fit the entire range \(\delta \in (0,1)\), they would complicate the analysis. In Fig.~\ref{fig:piecewise linear approximation of x_delta}, we plot the actual \(x_{\delta}(w)\) compared to our approximate expression \(\hat{x}_{\delta}(w)\) for \(\delta=0.2:0.1:0.9\). It is shown that very good alignment is achieved especially for \(w\leq 10\), while the Mean Squared Error (MSE) for all curves corresponding to different \(\delta\) is on the order of or less than \(10^{-2}\), as shown in the figure. 

It follows from~\eqref{eq: Dr} that smaller \(x_{\delta}(w)\) results in smaller \(\Delta^{\pm}_{\delta}(r,w)\) values, implying narrower beams. Our approximation \eqref{eq: piecewise linear fitting of x_delta(w)} reveals that, for \(w < 2.3\), \(x_{\delta}(w)\) decreases with \(w\) and becomes increasingly smaller compared to \(x_{\delta}(0)\), indicating increasingly directive beams in terms of beam depth limits. For \(w > w_0\), \(x_{\delta}(w)\) increases, surpassing \(x_{\delta}(0)\) at approximately \(w = 3.1\) for all \(\delta\), which can be found by setting \(b_0(\delta) + b_1(\delta)w=0\). Beyond this point, the beam progressively widens, losing its near-field focusing capability. This finding can also be validated from Fig.~\ref{fig: Lemma validation G(r Delta r)}, where it can be seen that the curve with \(w=4>3.1\) is less steep compared to the lossless curve \(w=0\), while \(w=1\) and \(w=2\) were steeper.

The phenomenon where attenuation up to \(w \cong3.1\) results in shorter beam depth limits than a lossless aperture, arises because the maximum gain is significantly lower in lossy scenarios. For instance, with \(w = 2\) in an \(N_e=200\) element array, the squared efficiency is \(\eta^2 = 0.06\) in \eqref{eq: Gain for perfect alignment}, meaning the maximum gain is only \(6\%\) of that of a lossless antenna. Consequently, while a percentage of the maximum gain is lost more rapidly, the absolute loss remains significantly lower.
 
In Fig.~\ref{fig: Beam depth limits wrt w}, we present the beam depth limits \(\Delta_{\delta}^{\pm}(r,w)\) for \(\delta = 0.9\) as a function of \(w\), using \eqref{eq: Dr} and our fitted model \(\hat{x}_{\delta}(w)\). Despite the approximations involved, the results are satisfactory: the computed \(\Delta^{\pm}_{0.9}(r,w)\) values yield relative beamforming gains (illustrated with the blue and red dashed lines) close to \(0.9\), with a maximum discrepancy of \(0.915\), or less than a \(2\%\) relative error. Notably, at \(w = 10.3\), \(r = \frac{d_e^2 N_e^2 \pi \sin^2(\theta)}{\lambda \hat{x}^2_{0.9}(w)}\), which aligns with the turning point where the red dashed line depicting \(\bar{G}(r+\Delta^{+}_{0.9}(r,w),\phi,\theta)\) begins to diverge beyond \(0.9\), as predicted in Section~\ref{subsec: Beam depth derivation}. After this point, \(\hat{x}_{0.9}(w)\) increases, leading to \(r > \frac{d_e^2 N_e^2 \pi \sin^2(\theta)}{\lambda \hat{x}^2_{0.9}(w)}\), meaning that no \(\Delta_{0.9}^{+}(r,w)\) beam depth limit exists. Additionally, as expected, for \(w > 3.1\), the beam depth limits \(\Delta^{\pm}_{0.9}(r,w)\) begin to become wider compared to the lossless limits, i.e., \(\Delta^{\pm}_{0.9}(r,w)\geq \Delta^{\pm}_{0.9}(r,0)\). Lastly, the slight peak of \(\bar{G}(r\pm \Delta^{\pm}_{0.9}(r,w),\phi,\theta)\) at \(w_0=2.3\) arises from the discontinuity error of \(\hat{x}_{0.9}(w)\), while the actual \(x_{0.9}(w)\) would exhibit a smoother transition maintaining the same trend.

\begin{figure}
    \centering
    \includegraphics[width=\columnwidth]{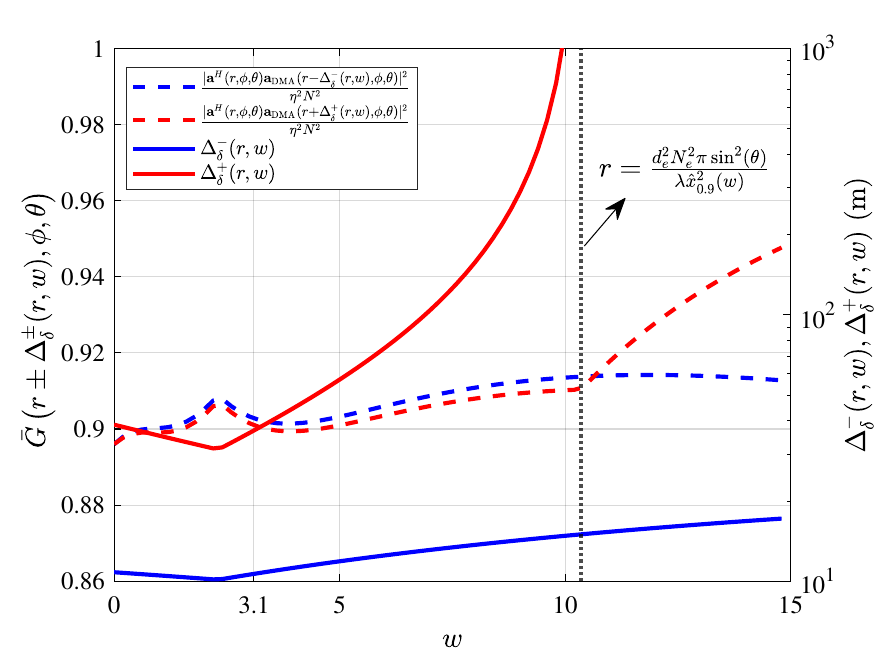}
    \caption{Right vertical axis: beam depth limits \(\Delta^{\pm}_{\delta}(r,w)\) with respect to \(w\) computed via \eqref{eq: Dr}; in their computation, our approximate function \(\hat{x}_{\delta}(w)\) was used to showcase its efficiency. Left vertical axis: the relative beamforming gain values at the beam depth limits. The DMA parameters are the same as in Fig.~\ref{fig: Lemma validation G(r Delta r)}, while \(r=30\), \(\phi=\pi/3\), \(\theta=\pi/3\), and \(\delta = 0.9\). At \(w=10.3\), \(r={N_e^2 d_e^2 \pi \sin^2(\theta)}/({\lambda \hat{x}_{\delta}^2(w)})\), which justifies that \(\Delta^{+}_{\delta}(r,w) \to \infty\) and, consequently, \(\bar{G}(r+\Delta^{+}_{\delta}(r,w),\phi,\theta)\) steadily diverges from \(\delta = 0.9\).}
    \vspace{-4mm}
    \label{fig: Beam depth limits wrt w}
\end{figure}

\vspace{-1.6mm}
\section{Conclusion}
\vspace{-1.4mm}
In this paper, we analyzed the impact of DMA losses on near-field beamforming performance. We derived analytical expressions for the beamforming gain under range misalignments, and established beam depth limits under varying attenuation conditions. Our results revealed that, while the losses reduce the maximum achievable gain, there is a threshold up to which increasing attenuation leads to steeper beams in terms of beam depth analysis. Beyond that point, the beam broadens, and the DMA progressively loses its near-field focusing capability. These findings provide valuable insights into our understanding of lossy DMA-based XL antenna arrays regarding near-field performance.
\vspace{-1.5 mm}
\bibliographystyle{IEEEtran}
\bibliography{references}

\end{document}